\documentclass[runningheads]{llncs}
\usepackage[T1]{fontenc}
\usepackage{graphicx}
\usepackage{color}
\usepackage{xspace}
\usepackage{mdframed}
\usepackage{amsmath,amssymb}

\usepackage{quoting,xparse}

\usepackage{tikz}
\usetikzlibrary{circuits.logic.US}

\usepackage{cite,enumitem}
\usepackage[linesnumbered,ruled]{algorithm2e}

\usepackage{hyperref}
\usepackage{mathtools}
\usepackage[binary-units]{siunitx}

\usepackage{listings}
\usepackage{pgfplots}
\pgfplotsset{width=5.75cm, compat=1.17}
\usepackage{todonotes}
\newcommand{\todoi}[1]{\todo[inline]{#1}}

\def\orcid#1{\smash{\href{http://orcid.org/#1}{\protect\raisebox{-1.25pt}{\protect\includegraphics{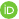}}}}}

\def\specfontfamily#1{\mathcal{#1}}
\def\spec{{\specfontfamily{S}}}

\def\<#1>{\langle#1\rangle}

\newcommand{\linGB}{\textsc{MultiLinG}\xspace}
\newcommand{\abc}{\textsc{ABC}\xspace}

\newcommand{\msolve}{\textsc{msolve}\xspace}
\newcommand{\dyposub}{\textsc{DyPoSub}\xspace}
\newcommand{\amulettwo}{\textsc{AMulet2}\xspace}
\newcommand{\teluma}{\textsc{TeluMA}\xspace}

\newcommand{\dynphaseorderopt}{\textsc{DynPhaseOrderOpt}\xspace}
\newcommand{\dpoo}{\textsc{DPOO}\xspace}

\newcommand{\ZZ}{\mathbb{Z}}
\newcommand{\KK}{\mathbb{K}}
\newcommand{\NN}{\mathbb{N}}
\newcommand{\BB}{\mathbb{B}}
\newcommand{\QQ}{\mathbb{Q}}
\DeclareMathOperator{\lt}{lt}
\DeclareMathOperator{\lm}{lm}
\DeclareMathOperator{\lc}{lc}
\DeclareMathOperator{\dist}{dist}
\DeclareMathOperator{\tail}{tail}

\DeclareMathOperator{\size}{size}
\DeclareMathOperator{\LEX}{lex}
\newcommand{\llex}{\mathrel{\prec_{\LEX}}}
\DeclareMathOperator{\DRL}{drl}
\newcommand{\ldrl}{\mathrel{\prec_{\DRL}}}
\newcommand\Ffour{\textsc{\texorpdfstring{F\textsubscript{4}}{F4}}\xspace}
\DeclareMathOperator\lin{lin}
\newcommand\ideal[1]{\langle#1\rangle}
\DeclareMathOperator\ext{ext}
\DeclareMathOperator\init{init}

\begin{document}
\title{Extracting Linear Relations from Gröbner Bases for Formal
  Verification of And-Inverter Graphs}
\titlerunning{Linear Gröbner Basis Polynomials for AIG Verification}
	%
\author{Daniela Kaufmann\inst{1}\,\orcid{0000-0002-5645-0292} \and J\'er\'emy Berthomieu\inst{2}\,\orcid{0000-0002-9011-2211}  }
\authorrunning{D. Kaufmann and J. Berthomieu}
	%
\institute{TU Wien, Austria \quad\texttt{daniela.kaufmann@tuwien.ac.at} \and Sorbonne Universit\'e, CNRS, LIP6, France \quad\texttt{jeremy.berthomieu@lip6.fr}\\}
\maketitle 
\setcounter{footnote}{0}
\begin{abstract}
  Formal verification techniques based on computer algebra have
  proven highly effective for circuit verification.
  The circuit, given as an and-inverter graph, is encoded using
  polynomials that automatically generate a Gröbner
  basis with respect to a lexicographic term ordering.
  Correctness of the circuit is derived by computing the
  polynomial remainder of the specification.
  However, the main obstacle is the monomial blow-up
  during the reduction, as the degree can increase.

  In this paper, we investigate an orthogonal approach and focus the
  computational effort on rewriting the  Gröbner basis itself. Our goal
  is to ensure the basis contains linear polynomials that can be
  effectively used to rewrite the linearized specification. We first
  prove the soundness and completeness of this technique and then
  demonstrate its practical application. Our implementation of this
  method shows promising results on benchmarks related to multiplier
  verification.

\end{abstract}
\section{Introduction}

Formal verification techniques based on algebraic reasoning have
emerged as highly effective tools for verifying hardware, particularly
in the context of verifying arithmetic circuits on the gate-level.
As digital systems become more
complex, ensuring the correctness of such circuits is paramount,
especially in safety-critical applications like cryptography and
signal processing to prevent a repetition of infamous failures, such as
the Pentium FDIV bug~\cite{Sharangpani1994StatisticalAO}.  Established
methods based on satisfiability solving (SAT)~\cite{Biere-SAT-Competition-2016-benchmarks}, or binary decision diagrams (BDDs)~\cite{ChenBryant-DAC95}
often struggle with the complex non-linear structure of arithmetic
circuits.
In contrast, formal verification techniques based on theorem provers~\cite{Temel-TACAS24} or computer algebra, specifically those leveraging Gröbner bases, offer an
effective alternative and have made significant progress in recent
years~\cite{KonradScholl-FMCAD24,KaufmannBiereKauers-FMCAD19,MahzoonGrosseSchollDrechsler-DATE20,KonradSchollMahzoonGrosseDrechsler-FMCAD22}.

In the algebraic method, circuits are given as and-inverter graphs
(AIG)~\cite{Kuehlmann-TCAD2002}. The graph is encoded as a set of
polynomials, which are sorted according to
a lexicographic term ordering, where for each gate in the circuit the
output variable is always greater than the input variables of the
gate. Hence, the leading terms of the polynomial
equations consist of single variables that are mutually disjoint.
This property is called \emph{unique monic leading term} (UMLT)
in~\cite{KaufmannBiereKauers-FMCAD19}.

If such an ordering is chosen, the polynomials automatically form a
Gröbner basis~\cite{Buchberger65}.  Informally said, a Gröbner basis
is a mathematical construct that offers a decision procedure that
guarantees soundness and completeness
of the verification process.  The correctness of the circuit is
determined by computing the unique polynomial remainder of the
specification polynomial, which represents the intended functionality
of the circuit, modulo the Gröbner basis. The circuit fulfills the
specification if, and only if, the final remainder is
zero~\cite{KaufmannBiereKauers-FMSD19}.

However, a major practical obstacle in using a lexicographic term ordering is the
significant computational effort during the reduction process, as the degree can increase.
More precisely the size of the intermediate reduction results generally increases, since the tails of the polynomials in the Gröbner basis have a higher degree than their leading terms.
This is often deferred to as \emph{monomial
blow-up}.
A study in~\cite{MahzoonGrosseDrechsler-ICCAD18} showed
that the intermediate reduction results for
16-bit multipliers can have more than $10^6$ monomials.
To address this challenge, various preprocessing and
rewriting algorithms have been developed, which syntactically or
semantically analyze the input circuit to remove redundant information
from the polynomial encoding, ultimately optimizing the reduction
process and improving the efficiency of the verification.

\paragraph{Related Work. }
Advanced reduction engines designed for the automatic algebraic verification of
 multipliers given as AIGs are implemented in
 tools such as \dynphaseorderopt~\cite{KonradScholl-FMCAD24},
\dyposub~\cite{MahzoonGrosseSchollDrechsler-DATE20}, and
\amulettwo~\cite{KaufmannBiereKauers-FMCAD19,KaufmannBiere-TACAS21},
including its variant
\teluma~\cite{KaufmannBeameBiereNordstrom-DATE22}.
In~\cite{KaufmannBiereKauers-FMCAD19, KaufmannBiere-TACAS21}
SAT solving is used to rewrite certain parts of the
multiplier before applying an incremental column-wise verification
algorithm. In a follow-up
work~\cite{KaufmannBeameBiereNordstrom-DATE22} the usage of the
external SAT solver could be removed by using a sophisticated
algebraic encoding that also takes the polarity of literals into
account.  These techniques have been further enhanced by
parallelization~\cite{LiuLiaoHuangZhenYuan-DATE24} and equivalence
checking-based verification~\cite{LiLiYuFujitaJiangHa-TCAD24}.

In~\cite{MahzoonGrosseSchollDrechsler-DATE20}, the authors present a
dynamic rewriting approach. They decide on
the reduction order on the fly and backtrack if the size of
intermediate reduction results exceeds a predetermined threshold.  In~\cite{KonradScholl-FMCAD24}, the authors revisit and
improve upon~\cite{MahzoonGrosseSchollDrechsler-DATE20} by incorporating mixed signals in their encoding.

While all of the discussed approaches employ various preprocessing and
rewriting techniques, they share a common characteristic: they all
rely on a lexicographic term ordering. \emph{None of the related works have explored alternative term orderings, such as those that prioritize degree-based sorting to limit the degree during the reduction process.}


\paragraph{Our contribution. }

In this paper, we propose an alternative, orthogonal strategy that
shifts the focus of the computational effort from rewriting the
specification to rewriting the Gröbner basis itself.
We impose a different term ordering that takes their degree into account.
The approach is based on the following observation:

\begin{quoting}[
  indentfirst=true,
  leftmargin=\parindent,
  rightmargin=\parindent]\itshape
  If the specification polynomial is linear, a Gröbner basis with
  respect to a degree reverse lexicographic term ordering contains
  linear polynomials that suffice to derive correctness of the
  circuit.
\end{quoting}

Our first contribution is to derive the theoretical foundations of
this observation, including a technical theorem that proves its
soundness and completeness.

However, the computation of a single Gröbner basis for the whole
circuit is practically infeasible due to the large number of variables and more importantly the
\emph{degree of the underlying ideal}.  Our second contribution is a
practical algorithm that splits the computation of the Gröbner basis
into multiple smaller more manageable sub-problems.  We evaluate our
approach on a set of benchmarks for multiplier verification. The
experimental results are promising and indicate that our approach
offers a valuable addition to existing algebraic
verification techniques.

The remainder of the paper is organized as follows. In
Section~\ref{sec:prelim} we introduce the necessary preliminaries. In
Section~\ref{sec:grevlex_verification} we show the theory of our
approach and prove the soundness and
completeness. We present a practical verification algorithm in Section~\ref{sec:practical_grevlex_verification}, and discuss its implementation
and the experimental evaluation in Section~\ref{sec:experiments}
before we conclude the paper in Section~\ref{sec:conclusion}.

\section{Preliminaries} \label{sec:prelim}

In the first part of the preliminaries, Section~\ref{ssec:alg}, we
introduce the theory of Gröbner bases
following~\cite{Buchberger65,buchberger10,CoxLittleOShea-Book07} and
discuss key properties that are important for our approach. In the
second part, Section~\ref{ssec:aig}, we present the necessary background
on AIGs and how we can encode these graphs using polynomial equations.

\subsection{Gröbner Basis}\label{ssec:alg}

\begin{definition}[Term, Monomial, Polynomial, see~{\cite[Chap.~2,
    Sec.~2, Def.~7]{CoxLittleOShea-Book07}}]
  Let $X = (x_1,\dots,x_n)$ be a set of variables and $\KK$ be a field.  A
  \emph{monomial} is a product of the form
  $x_1^{e_1}\cdots x_n^{e_n}$, with exponents $e_1,\dots,e_n \in \NN_0$.
  The set of all monomials is represented by
  $[X]$.  A \emph{term} is a monomial multiplied by a constant,
  written as $\alpha x_1^{e_1}\cdots x_n^{e_n}$ with $\alpha \in \KK$.
  A \emph{polynomial} $p$ is a finite sum of such terms. We denote the
  number of terms in $p$ by $\size(p)$.
\end{definition}

Throughout this section let $\KK[X]=\KK[x_1,\dots,x_n]$ denote the
ring of polynomials in variables $x_1,\dots,x_n$ with coefficients in
the field $\KK$. We write polynomials in their canonical form. That
is, monomials with equal monomials are merged by adding their
coefficients; and terms with coefficients equal to zero are
removed.

\begin{definition}[Degree]
  The degree of a monomial $\sigma= x_1^{e_1}\cdots x_n^{e_n}$ is the sum
  of its exponents, i.e., $\deg(\sigma) = |\sigma| = \sum_{i=1}^{n}{e_i}$.  The
  degree of a polynomial is the maximum degree of its terms.
\end{definition}

The terms within a polynomial are sorted according to a total order to ensure a consistency for algebraic operations.

\begin{definition}[Monomial Order]
  A monomial order is a total order $\prec$ such that for all distinct
  monomials
  $\sigma_1,\sigma_2$ we have (i) $\sigma_1 \prec \sigma_2$ or
  $\sigma_2 \prec \sigma_1$, (ii) every non-empty set of monomials has a
  smallest element and
  (iii)~$\sigma_1\prec\sigma_2\Rightarrow\tau\sigma_1\prec\tau\sigma_2$
  for any term $\tau$.
\end{definition}

\begin{definition}[Lexicographic Order, see~{\cite[Chap.~2, Sec.~2,
    Def.~3]{CoxLittleOShea-Book07}}]
  Let $\sigma_1= x_1^{u_1}\cdots x_n^{u_n}$ and  $\sigma_2=
  x_1^{v_1}\cdots x_n^{v_n}$ be two monomials. We say that $\sigma_1
  \llex \sigma_2$,
  if there exists an index $i$ such that
  with $u_j = v_j$ for all $1\leq j<i$, and $u_i<v_i$.
\end{definition}

\begin{definition}[Degree Reverse Lexicographic Order,
  see~{\cite[Chap.~2, Sec.~2, Def.~6]{CoxLittleOShea-Book07}}]
  Let $\sigma_1= x_1^{u_1}\cdots x_n^{u_n}$ and  $\sigma_2=
  x_1^{v_1}\cdots x_n^{v_n}$ be two monomials. We say that $\sigma_1
  \ldrl \sigma_2$,
  if $|\sigma_1|<|\sigma_2|$ or if $|\sigma_1|=|\sigma_2|$ and
  there exists an index~$i$ such that
  $u_j = v_j$ for all $i<j\leq n$, and $u_i>v_i$.
\end{definition}

Since every polynomial $p \in \KK[X]$ contains only a finite number of
monomials, and these terms are sorted according to a fixed total order, we can
identify the largest monomial in $p$.  This is referred to as the
\emph{leading monomial} of $p$ and denoted as $\lm(p)$.  If
$p=c\tau + \cdots$ and $\lm(p)=\tau$, then $\lc(p)=c$ is called the
\emph{leading coefficient} and $\lt(p)=\lc(p)\lm(p)=c\tau$ is called the
\emph{leading term} of~$p$. The \emph{tail} of $p$ is defined by
$\tail(p) = p-\lt(p)$.

\begin{definition}[Ideal]\label{def:ideal}
  A nonempty subset $I\subseteq \KK[X]$ is called an \emph{ideal} if
  \[
    \forall\,u,v\in I: u+v\in I \quad\text{ and }\quad \forall\,w\in
    \KK[X]\ \forall\,u\in I: wu\in I.
  \]
  If $I\subseteq \KK[X]$ is an ideal, then a set
  $G=\{g_1,\dots,g_m\}\subseteq \KK[X]$ is called a \emph{basis} of~$I$ if $I=\{h_1g_1+\cdots+h_mg_m\mid h_1,\dots,h_m\in \KK[X]\}$,
  i.e., if $I$ consists of all the linear combinations of $g_i$
  with polynomial coefficients.  We denote this by
  $I= \ideal{G}$ and say $I$ is generated by $G$.
\end{definition}

An ideal $I=\ideal{G}\subseteq \KK[X]$ can be interpreted as an equational theory,
 where the basis $G=\{g_1,\dots,g_m\}$ serves as the
set of axioms. The ideal $I = \ideal{G}$ consists of precisely
those polynomials $f$ for which the equation $f=0$ can be derived
from the axioms $g_1=\cdots=g_m=0$ through repeated application of
the rules $u=0\land v=0\Rightarrow u+v=0$ and $u=0\Rightarrow wu=0$.

To check whether a polynomial $f \in \KK[X]$ is contained in an ideal
$I$, we want to solve the so-called \emph{ideal membership problem:}
Given a polynomial $f\in \KK[X]$ and an ideal
$I= \ideal{G} \subseteq \KK[X]$, determine if $f \in I$.

\begin{definition}[Remainder]
  The process of finding a remainder with respect to a set of
  polynomials $G$ is equal to computing the remainder of a polynomial
  division, but extended to multiple divisors, until no further
  division is possible.  The result is a polynomial that represents
  the equivalent class modulo the ideal generated by $G$. We write
  $p \rightarrow_G g$ to denote that $g$ is the polynomial remainder
  of $p$ modulo $G$ and we also say ``$p$ is reduced by $G$''.
\end{definition}

In general, an ideal $I$ has many bases that generate $I$.  We
are particularly interested in bases with certain structural
properties that allow to uniquely answer the ideal membership
problem. Such bases are called \emph{Gröbner
  bases}~\cite{Buchberger65}.
\begin{lemma}[see~{\cite[Chap.~2, Sec.~5, Cor.~6]{CoxLittleOShea-Book07}}]
  Every ideal $I \subseteq \KK[X]$ has a Gröbner basis w.r.t.\ a fixed
  total order.
\end{lemma}

Given an arbitrary basis of an ideal, a Gröbner basis can be computed
using Buchberger's algorithm that repeatedly computes so-called
S-Polynomials. These S-Polynomials are reduced by the polynomials that
are already in the current basis, i.e., calculating the remainder of
polynomial division, and non-zero remainders are added to the ideal
basis. These steps are repeated until the basis is saturated. If all
S-Polynomials reduce to zero the set of ideal generators is a Gröbner
basis~\cite{Buchberger65}.
Generally, Buchberger-like algorithms for computing
Gröbner bases, such as Buchberger's seminal
algorithm~\cite{Buchberger65} or Faugère's~\Ffour~\cite{Faugere1999}
algorithm, have a worst-case time complexity double exponential in the
number of variables, because of the size of the
output~\cite{MayrM1982}. Still, in practice, these algorithms behave
in
general way better for $\ldrl$ than for other monomial orders, such as $\llex$.

We will not introduce this process more
formally, as we will treat the computation of a Gröbner basis as a
black-box technique in our approach.
The following properties are more important for us.

\begin{lemma}[see~{\cite[Chap.~2, Sec.~6, Prop.~1]{CoxLittleOShea-Book07}}]\label{lemma:gbx}
  If $G = \{g_1,\ldots,g_m\}$ is a Gröbner basis, then every
  $f\in\KK[X]$ has a unique polynomial remainder $r$ with respect
  to~$G$. Furthermore, it holds that $f-r \in \ideal{G}$, which implies that
  $f$ is contained in the ideal $I=\ideal{G}$ if, and only if, $f \rightarrow_G 0$.
\end{lemma}

Depending on the information one seeks, some Gröbner bases are more
useful than others. Gröbner bases w.r.t.~$\llex$ are the tool of choice
for solving polynomial systems but are, in general, more expensive to
compute than degree-based Gröbner bases. Yet, \emph{change of order}
algorithms, such as the seminal \textsc{FGLM} one~\cite{FaugereGLM1993} can
convert a Gröbner basis into another one for different order. In our setting of verifying AIG the complexity would be
in $O(n 2^{3 n})$, where $n$ is the number of input variables of the
AIG. Hence, for large $n$, this is impractical.
Variants of \textsc{FGLM}
exploiting the structure of the input and output Gröbner bases under
some genericity assumptions exist,
we can mention~\cite{BerthomieuNS2022,FaugereGHR2014,FaugereM2017,NeigerS2020},
but they are mostly designed for solving polynomial systems. As a
consequence, they consider the input Gröbner basis to be for
a degree-based order, such as $\ldrl$, and the output Gröbner basis to
be for
$\llex$.

\subsection{And-Inverter Graphs} \label{ssec:aig}

An \emph{and-inverter graph} (AIG)~\cite{Kuehlmann-TCAD2002} is a
special case of a directed acyclic graph (DAG).  They
are useful tools to represent Boolean functions and logic circuits and
provide a compact and efficient way to describe logical expressions.

\begin{definition}[AIG]
  An AIG operates over Boolean variables.  Every node
  expresses a logical conjunction between its two input variables,
  which are depicted by incoming edges in the lower part of
  the node. We distinguish two types of inputs, primary inputs (of the graph)
  and intermediate nodes.  Outputs of the node are represented by an
  edge in the upper half. If an edge is marked, it indicates that the
  variable is negated.
\end{definition}

\begin{definition}[Specification]
  The specification of an AIG is a polynomial equation
  $\spec \in \KK[X]$ that relates the outputs of an AIG to
  its primary inputs.
\end{definition}

Together with the specification polynomial, we fix the polynomial ring $\KK[X]$
of the encoding. Although the nodes in an AIG compute logical
conjunction over Boolean variables, the specification can encode
richer relations. Hence, the encoding is not restricted to the Boolean
ring $\BB[X]$, but may include different coefficient domains, such as
integers or rationals.

\begin{definition}[Gate Polynomials]\label{def:gatepoly}
  Each node in an AIG can be encoded by a corresponding polynomial
  equation that models the logical conjunction.  Nodes in an AIG
  raise
  four
  types of equations, depending if either none,
  the first, the second,
  or both inputs are negated. Let $g$ be an AIG node with inputs
  $a, b$:
  \[
    \arraycolsep=10pt
    \begin{array}{lcr}
      \mbox{Gate constraint} &            & \mbox{Gate polynomial}\\
      g = a \land b						& \Rightarrow &		g - ab = 0 \\
      g = \neg a \land b      & \Rightarrow &   g-(1-a)b =  g + ab - b  = 0\\
      g = a \land \neg b      & \Rightarrow &   g-a(1-b) =  g + ab - a  = 0\\
      g = \neg a \land \neg b & \Rightarrow &   g -(1-a)(1-b) =  g - ab + b +a -1 = 0  \end{array}\]
\end{definition}

The correctness of the encoding can easily be checked by truth tables.
Furthermore, observe that the degree of the gate
polynomials is always two.

\begin{definition}[Boolean Input Polynomial]\label{def:boolpoly}
  For every primary input $a_i$ of the AIG we define a corresponding
  \emph{Boolean input polynomial} $a_i(a_i-1) = a_i^2-a_i = 0$ that
  encodes that the variable can only take the values 0 and 1.
\end{definition}

As we will only consider polynomial equations with right hand side
zero, we will from now on shorten our notation and write ``$f$'' instead of ``$f =0$''.

\begin{figure*}[tb]
  \centering
  \begin{minipage}{0.35\textwidth}
    \includegraphics[width=.99\textwidth]{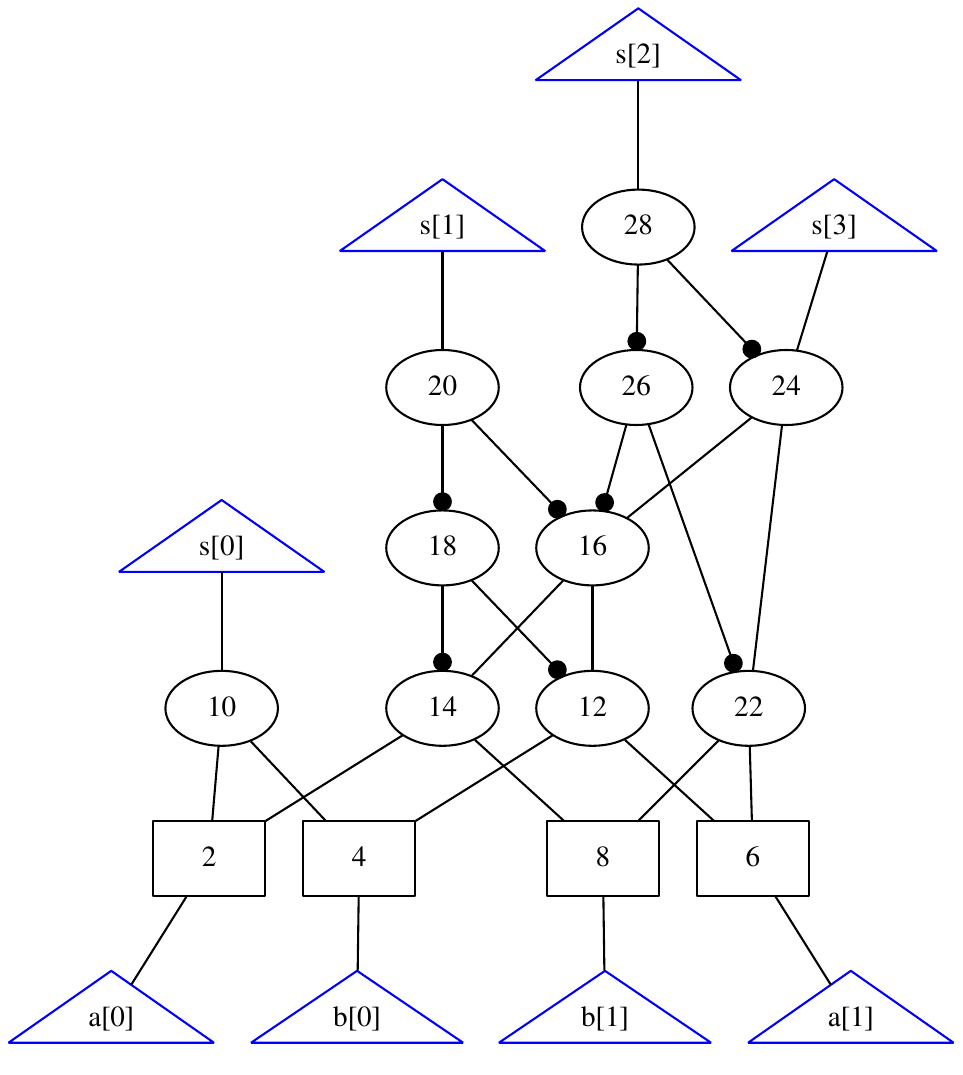}
  \end{minipage}\hfill
  \begin{minipage}{0.64\textwidth}
    \begin{center}
      \scriptsize $
      \begin{array}{c@{\qquad}l@{\qquad}l}
        \mbox{Index} &\mbox{Gate Polynomial} &\mbox{Gate constraint} \\
        g_0 &   s_3 -                             \ell_{24}  &   s_3  = \ell_{24}                        \\
        g_1 &   s_2 -                             \ell_{28}  &   s_2  = \ell_{28}                        \\
        g_2 &\ell_{28} - \ell_{26}\ell_{24} + \ell_{26} + \ell_{24} - 1  & \ell_{28} = \neg \ell_{26} \land \neg \ell_{24} \\
        g_3 &\ell_{26} - \ell_{22}\ell_{16} + \ell_{22} + \ell_{16} - 1  & \ell_{26} = \neg \ell_{22} \land \neg \ell_{16} \\
        g_4 &\ell_{24} - \ell_{22}\ell_{16}                        & \ell_{24} = \ell_{22} \land \ell_{16}           \\
        g_5 &\ell_{22} - b_1a_1                              & \ell_{22} = b_1 \land a_1                 \\
        g_6 &   s_1 -                             \ell_{20}  &   s_1  = \ell_{20}                        \\
        g_7 &\ell_{20} - \ell_{18}\ell_{16} +\ell_{18} + \ell_{16} - 1  & \ell_{20} = \neg \ell_{18} \land \neg \ell_{16} \\
        g_8 &\ell_{18} - \ell_{14}\ell_{12} + \ell_{14} + \ell_{12} - 1  & \ell_{18} = \neg \ell_{14} \land \neg \ell_{12} \\
        g_9 &\ell_{16} - \ell_{14}\ell_{12}                        & \ell_{16} = \ell_{14} \land \ell_{12}           \\
        g_{10} &\ell_{14} - b_1a_0                              & \ell_{14} = b_1 \land a_0                 \\
        g_{11} &\ell_{12} - b_0a_1                              & \ell_{12} = b_0 \land a_1                 \\
        g_{12} &   s_0 -                             \ell_{10}  &   s_0  = \ell_{10}                        \\
        g_{13} &\ell_{10} - b_0a_0                              & \ell_{10} = b_0 \land a_0                 \\
      \end{array}
      $\end{center}
  \end{minipage}

  \vspace{2ex}
 \scriptsize
  Boolean Input Polynomials: $a_1^2-a_1, a_0^2-a_0, b_1^2-b_1, b_0^2-b_0$ \\
  Spec $\spec$:
  $\sum_{i=0}^3 2^is_i = (\sum_{i=0}^1 2^ia_i)(\sum_{i=0}^1 2^ib_i) = 8s_3+4s_2+2s_1+s_0 - 4a_1b_1 - 2a_1b_0-2a_0b_1-a_0b_0$
  \caption{AIG and polynomial encoding of a 2-bit multiplier
    in the ring $\QQ[X]$.}
  \label{fig:btor2}
\end{figure*}

\begin{example}\label{example:gate}
  Figure~\ref{fig:btor2} shows an AIG representing a 2-bit
  multiplier.
  We denote the primary
  inputs by $a_0, a_1, b_0, b_1$ and outputs by
  $s_0, s_1, s_2, s_3$.  The internal nodes are denoted by
  $\ell_i$, with subscript $i$ corresponding to the number of the respective AIG node. The right hand side of Figure~\ref{fig:btor2} lists
  the gate constraints as well as the corresponding gate polynomials, which are derived using the encoding presented in Def.~\ref{def:gatepoly}. We furthermore list the Boolean input polynomials (Def.~\ref{def:boolpoly}) and the specification polynomial $\spec \in \QQ[X]$,
  which relates that $S = A \cdot B $, for
  $S = \sum_{i=0}^3 2^is_i$, $A = \sum_{i=0}^1 2^ia_i$, and
  $B = \sum_{i=0}^1 2^ib_i$.
\end{example}

\section{Verification using Degree Reverse Lexicographic
  Order}\label{sec:grevlex_verification}

In this section we will lay out the theoretical foundation of our
proposed approach for extracting linear relations from the Gröbner
basis that is used for reduction.

Existing algebraic verification techniques
for acyclic graphs encode the circuit as a polynomial using a
lexicographic term ordering where the variables are sorted according
to a reverse topological term ordering
(RTTO)~\cite{LvKallaEnescu-TCAD13}. This has the benefit that due to
repeated application of Buchberger's product criterion,
see~\cite[Chap.~2, Sec.~10, Prop.~1]{CoxLittleOShea-Book07}, the set
of gate polynomials together with the Boolean input polynomials
automatically form a Gröbner basis~\cite{KaufmannBiereKauers-FMSD19}.
Since the leading terms of the gate polynomials consist of one single
variable, polynomial division
comes down  to substitution.  The variables in the
specification are substituted by the corresponding tails of the gate
polynomials until no further rewriting is possible.
The graph fulfills
its specification if, and only if, the final result is zero.

Generally, this implies that  the
degree of the intermediate reduction results increases, since the tails of the
gate polynomials have a higher degree than their linear leading terms.
Substituting those variables in non-linear monomials has the potential to lead to a monomial blow-up during the reduction.

\begin{example}\label{example:gatelex}
  Consider again the polynomials of
  Example~\ref{example:gate}.  Initially $\size(\spec) = 8$
  and $\deg(\spec) = 2$. After four rewriting steps
  we have the following intermediate reduction result: $\spec
  \rightarrow_{\{g_1,g_2,g_3,g_4\}}
  4\ell_{24}\ell_{22}\ell_{16}-4\ell_{24}\ell_{22}-4\ell_{24}\ell_{16}+8\ell_{24}-4\ell_{22}\ell_{16}+4\ell_{22}+2s_{1}+4\ell_{16}+s_0
  - 4a_1b_1 - 2a_1b_0-2a_0b_1-a_0b_0$ which has degree 3 and
  consists of 13 monomials.
\end{example}

 We will now
impose a different ordering on the set of gate polynomials
that takes the degree of the polynomials into account.  That
is, we compute a Gröbner basis based on the degree reverse lexicographic
monomial ordering, where the monomials in a polynomial are first
sorted according to their degree.

Our approach is based on the following
result that we prove in Theorem~\ref{thm:linGB}:
\emph{If
  the specification polynomial is linear, then the ideal
  membership of the specification can be decided using only
  the linear polynomials of the Gröbner basis.}

We linearize the specification by replacing all
non-linear monomials~$\sigma_i$ in~$\spec$ with new extension variables~$t_i$.
For each
replacement we generate a new polynomial constraint
$t_i - \sigma_i$ and add it to the set of gate
polynomials. The idea of linearization is, for instance, already
used in~\cite{LiewNordstroem-FMCAD20} in the context of cutting planes.

The next
lemma proves that if the specification~$\spec$ is
contained in the ideal generated by the gate polynomials, then
the linearized specification $\spec_{\lin}$ is contained in the
ideal generated by the gate polynomials and \emph{extension
  polynomials}.

\begin{lemma}\label{lemma:linspec}
  Let $p \in \KK[X]$, $I \subseteq \KK[X]$.  Let
  $\Sigma = \{t_i - \sigma_i \mid t_i \notin X \land \sigma_i \in p \land
  \deg(\sigma_i) > 1 \}$.  Let $p_{\lin}$ be
  the polynomial where every non-linear monomial of $p$ is
  replaced by a corresponding extension variable $t_i$. Then
  we have $p \in I$ if, and only if,
  $p_{\lin} \in I + \ideal{\Sigma}$.
\end{lemma}
\begin{proof}
  Let $p=\sum_{\sigma_i\in p}c_i \sigma_i$, where the $c_i$'s are in
  $\KK$. By definition, we can write
  $p_{\lin}=\sum_{\sigma_i\in p}c_i t_i$. Thus, $p_{\lin}
  = \sum_{\sigma_i\in p}c_i(t_i - \sigma_i)
  + \sum_{\sigma_i\in p}c_i\sigma_i$. By hypothesis, the first sum is
  in $\ideal{\Sigma}$ and the second one, which is $p$, only depends
  in variables~$X$. Therefore,
  if $p\in I$, then $p_{\lin}\in I+\ideal{\Sigma}$. Conversely, if
  $p_{\lin}\in I+\ideal{\Sigma}$, then $p\in
  (I+\ideal{\Sigma})\cap\KK[X]=I$ by construction of $I+\ideal{\Sigma}$.
\end{proof}

We now show soundness and completeness of our observation.
That is, if we
want to show ideal membership of a linear polynomial, the
Gröbner basis of the ideal contains a set of linear
polynomials $G_1$ that suffice for deriving the ideal
membership, all non-linear polynomials of the Gröbner basis can be
neglected. This will shift the computational difficulties from
the reduction to the Gröbner basis generation.

\begin{theorem}\label{thm:linGB}
  Let $p \in \KK[X]$ with $\deg(p) = 1$, $I \subseteq
  \KK[X]$ be an ideal. Let $G$ be a Gröbner basis of $I$ with respect
  to $\ldrl$ and let
  $G_1 = \{g \in G \mid \deg(g) \leq 1\}$.  We have
  $p \in I$ if, and only if, $p \rightarrow_{G_1} 0$. In
  particular, $p = \alpha_1g_1 + \cdots + \alpha_mg_m$ with
  $g_i \in G_1$, $\alpha_i \in \KK$.
\end{theorem}
\begin{proof}
  First, let us observe that if $G_1$ contains a non-zero constant polynomial,
  then $I=\KK[X]$ and $p$
  necessarily reduces to $0$ by $G_1$.

  We now assume that $G_1$ only contains polynomials of degree
  $1$. For $g\in G_1$, we write $g=\lt(g)+\tail(g)$. Because polynomials in $G_1$ are  ordered with respect to $\ldrl$, we have
  $\deg(\lm(g))=1$ and $\deg(\tail(g))\leq 1$. Since
  $\deg(\lm(p))=1$ the
  division algorithm for computing the reduction of $p$~by~$G$,
  see~\cite[Chap.~2, Sec.~3]{CoxLittleOShea-Book07}, will
  only select polynomials in $G$ whose leading monomials also have
  degree $1$, i.e.\ those in $G_1$. The reduction step will replace
  $p$ by $p-\alpha_i g_i=\tail(p)-\alpha_i\tail(g_i)$, for
  $\alpha_i\in\KK^*$ and some $g_i$,  which has degree less or
  equal to $1$.

  Since $p\in I$ if, and only if, $p\rightarrow_G 0$, we have
  $p\in I$ if, and only if, $p\rightarrow_{G_1} 0$.
\end{proof}

\emph{We emphasize that the theory, and in particular the result of Theorem~\ref{thm:linGB}, can be applied to general DAGs.}
The key property of the graph is that it must be \emph{acyclic}. If it has cycles one cannot canonically compare
the variables, hence it is not possible to derive a total term order and compute a Gröbner basis.

The conclusion of Theorem~\ref{thm:linGB} moreover shows that we can significantly
simplify the algorithm for checking the ideal membership of $\spec$. Instead
of repeated polynomial substitution, with potential non-linear
intermediate reduction results, we pick $g_i \in G_1$, such
that $\lm(g_i) = \lm(\spec)$, multiply $g_i$ by a constant $\alpha_i$
such that $\lc(\spec) = -\alpha_i\lc(g_i)$ and add those two
polynomials. Hence, we have replaced polynomial division by
linear polynomial operations.

Therefore, we can apply the following approach to verify
that an AIG fulfills its specification, see
Alg.~\ref{alg:lingbred}. We first encode the graph as a set of
polynomials $G_{\init}$ (line 1), and linearize the specification
(line~2)
as described in Lemma~\ref{lemma:linspec}. The set $G_{\ext}$ contains the extension polynomials.
In the next step we compute a Gröbner basis w.r.t.\ $\ldrl$
(line~3) and extract the linear
polynomials $G_1$ (line~4). We calculate the remainder of the
specification modulo the linear elements of the Gröbner basis
(lines~5--8) until no further reduction is possible and return whether the final result is zero.
The correctness of Alg.~\ref{alg:lingbred} follows from Theorem~\ref{thm:linGB}.

\begin{algorithm}[t]\small
  \SetKwInOut{Input}{Input} \SetKwInOut{Output}{Output}
  \Input{Circuit $C$ in AIG format, Specification polynomial
    $\spec$} \Output{Determine whether $C$ fulfills the
    specification} $G_{\init} \leftarrow \text{Gate-Polynomials}(C) \cup \text{Boolean-Input-Polynomials}(C)$\;
  $\spec_{\lin}, G_{\ext} \leftarrow$ Linearize($\spec$)\;
  $G_{\DRL} \leftarrow$
  Compute-$\ldrl$-Gröbner-Basis($G_{\init} \cup
  G_{\ext}$)\;
  $G_1 \leftarrow \{g \mid g \in G_{\DRL}
  \land \deg(g) \leq 1\}$\;
  \While{$\lm(\spec_{\lin}) \in \{\lm(g) | g \in
    G_1\}$}{
    $p_{\lin} \leftarrow g \in G_1$ such that $\lm(g) = \lm(\spec_{\lin})$\;
    \lIf{$\nexists\,p_{\lin}$}{\Return $\bot$}
    $\spec_{\lin} \leftarrow$  Linear-Reduce($\spec_{\lin}, p_{\lin}$)\;
  }
  \Return $\spec_{\lin} = 0$
  \caption{Linear Gröbner basis reduction}\label{alg:lingbred}
\end{algorithm}

\begin{example}\label{example:gatedlex}
  Consider again the AIG of Example~\ref{example:gate}.  First
  of all we define four extension variables $t_{ij}$ to encode
  the non-linear terms $a_ib_j$ for $i,j \in \{0,1\}$ and
  rewrite the specification to
  $8s_3+4s_2+2s_1+s_0 - 4t_{11} - 2t_{10}-2t_{01}-t_{00}$. The
four polynomial equations $ t_{ij} - a_ib_j$ are added to the
  set of gate polynomials and we compute a Gröbner basis w.r.t.\
  $\ldrl$. The full Gröbner basis consisting of 52
  polynomials is listed in Appendix~\ref{appendix:deggb-2mult}.

  Important for us are the first thirteen elements of the
  Gröbner basis, as those are the linear polynomials $G_1$:

  $
  \begin{array}{l@{=\quad}l@{\qquad\qquad}l@{=\quad}l}

    g_1 & \ell_{10}-t_{00} & g_8 & \ell_{22}-t_{11} \\
    g_2 & s_0-\ell_{10}    & g_9 & \ell_{24}-\ell_{16} \\
    g_3 & \ell_{12}-t_{10} & g_{10} & \ell_{26}-\ell_{24}+\ell_{16}+t_{11}-1 \\
    g_4 & \ell_{14}-t_{01} & g_{11} & \ell_{28}+2\ell_{24}-\ell_{16}-t_{11} \\
    g_5 & \ell_{18}-\ell_{16}+t_{01}+t_{10}-1 & g_{12} & s_2-\ell_{28} \\
    g_6 & \ell_{20}+2\ell_{16}-t_{01}-t_{10} & g_{13} & s_3-\ell_{24} \\
    g_7 & s_1-\ell_{20} \\
  \end{array}
  $

  We derive that $\spec \in \ideal{G_1}$
  as
  $\spec =
  8g_{13}+4g_{12}+4g_{11}+2g_7+2g_6+g_2+g_1$.
\end{example}

In practice, however line~3 of Alg.~\ref{alg:lingbred} turns
out to be a bottleneck, since computing a single $\ldrl$-Gröbner
basis does not scale for larger AIGs.

We have also seen in
Example~\ref{example:gate} that $39$ out of $52$ polynomials in
the computed Gröbner basis are non-linear. While these
polynomials are needed to compute the
full Gröbner basis of the ideal, they are not required for
solving the ideal membership problem of the linear specification.
Furthermore, from the $13$ linear polynomials, only $7$ are used to generate
the specification.  Hence, our generated Gröbner basis
contains redundant and/or useless information. We will now discuss a method to reduce
the overhead by computing local Gröbner bases.


\section{Locally extracting Linear
  Polynomials}\label{sec:practical_grevlex_verification}

The core idea of the optimized approach is to start from a
$\llex$-Gröbner basis and incrementally extract linear polynomials
from a smaller set of gate polynomials instead of computing a
single
full $\ldrl$-Gröbner basis
for the whole
input AIG.  The algorithm is outlined in
Alg.~\ref{alg:specred} and will be explained in more detail
throughout the remainder of this section.

In a nutshell, we first encode the circuit using a lexicographic term
ordering~(line~1) and linearize the specification
polynomial (line~2) with respect to the given circuit.  After
some preprocessing where we extract easily derivable
linear polynomials (line~3), we rewrite the specification by
generating linear polynomials on the fly (lines~4--9). We pick
the gate polynomial $p$ that has the same leading term as
the intermediate reduction result (line~5) and compute a
$\ldrl$-Gröbner basis for a sub-circuit of $C$ that includes
$p$ (line~6) to receive the linearized polynomial $p_{\lin}$ that
we use for reducing the specification (line~7).
Let us now go into more detail of every step.

\begin{algorithm}[t]\small
  \SetKwInOut{Input}{Input} \SetKwInOut{Output}{Output}
  \SetKwComment{Comment}{\qquad \qquad\qquad$\triangleright$}{}
  \SetKw{And}{and}

  \Input{Circuit $C$ in AIG format, Specification polynomial $\spec$}
  \Output{Determine whether $C$ fulfills the specification}
  $G_{\init} \leftarrow$ Row-Wise-RTTO-Polynomial-Encoding($C$)\;
  $\spec_{\lin},G_{\ext} \leftarrow$
  Linearize-Spec-wrt-AIG($\spec, G_{\init}$)\; $G \leftarrow$
  Preprocessing($G_{\ext}$) \Comment{See
    Section~\ref{ssec:preprocessing}}

  \While{$\lm(\spec_{\lin}) \in \{\lm(g) | g \in G\}$}{
    $p \leftarrow g \in G $ such that
    $\lm(g) = \lm(\spec_{\lin})$\;
    $p_{\lin} \leftarrow$
    Linearize-Single-Polynomial($p, G$) \Comment{See Section~\ref{ssec:linearrewr}}
    \lIf{$\nexists\,p_{\lin}$}{\Return $\bot$}
    $\spec_{\lin} \leftarrow$ Linear-Reduce($\spec_{\lin}, p_{\lin}$)\;
  }
  \Return $\spec_{\lin} = 0$
  \caption{Verification-via-Locally-extracting-Linear-Polynomials}\label{alg:specred}
\end{algorithm}

\paragraph{Encoding.}
The AIG is encoded using gate polynomials and Boolean input
polynomials as described in Definitions~\ref{def:gatepoly}
and~\ref{def:boolpoly} using a lexicographic term ordering.  We
choose a row-wise variable ordering that sorts variables based on
their distance to the inputs.  If nodes have an equal distance, we
sort according to the value of the AIG node.  For example, we would
sort the variables in Example~\ref{example:gate} as
$a_0 \llex b_0 \llex a_1 \llex b_1 \llex \ell_{10} \llex \ell_{12}
\llex \ell_{14} \llex \ell_{22} \llex s_{0} \llex \ell_{16} \llex
\ell_{18} \llex \ell_{20} \llex \ell_{24} \llex \ell_{26} \llex
\ell_{28} \llex s_1 \llex s_3 \llex s_2$.  In
this order the output variable of a gate is always greater than its
input variables, which automatically generates a $\llex$-Gröbner basis.

Theorem~4 in~\cite{KaufmannBiereKauers-FMSD19} has shown that we can
locally rewrite elements of the $\llex$-Gröbner basis without
jeopardizing the Gröbner basis property as long as the leading monomials remain the same.
We apply the same technique and locally rewrite gate polynomials from quadratic to linear
polynomials that will be used in the reduction.

\paragraph{Linearization of the Specification.}
Lemma~\ref{lemma:linspec} provides us with a methodology on how to
linearize the specification $\spec$ by introducing extension
variables to represent non-linear terms.  However, some of the terms
might already be contained in the polynomial encoding of the
circuit. For those terms we can simply use the corresponding leading
term in the specification.  We first swipe through the set of
gate polynomials and check whether the non-linear tail of a
gate polynomial is contained in the specification. If this is the
case, we replace the non-linear term by the corresponding leading
term.

For instance, in Example~\ref{example:gatedlex}
we have the gate polynomial $\ell_{22}-a_1b_1$. Hence, we do not
require the extension variable $t_{11}$ to linearize $\spec$.
This equality $\ell_{22} = t_{11}$ is also contained as
polynomial $g_8$ in the computed $\ldrl$-Gröbner basis.

All non-linear terms of $\spec$ that cannot be linearized using gate polynomials we
introduce extension variables as described in
Section~\ref{sec:grevlex_verification}.

\ \\
At this point, our encoding consists of a linear specification
polynomial and a set of quadratic gate polynomials and Boolean input polynomials that generate a
Gröbner basis w.r.t.\ a lexicographic term ordering. The following
subsections present how we linearize elements of the Gröbner basis.

\subsection{Preprocessing}\label{ssec:preprocessing}

The goal of preprocessing is to eliminate variables and derive linear polynomials in the
$\llex$-Gröbner basis that can be identified using simple
heuristics.
We employ three steps of rewriting, depicted in
Alg.~\ref{alg:preprocess}.

\begin{algorithm}[tb]\small
  \SetKwInOut{Input}{Input} \SetKwInOut{Output}{Output}
  \SetKw{And}{and}

  \Input{Set of poynomial encodings of gate constraints $G$}
  \Output{Rewritten Set of Polynomial encodings $G$}
  $G \leftarrow $ Merge-Nodes-with-Equal-Inputs($G$)\;
  $G \leftarrow $ Eliminate-Positive-Nodes($G$)\;
  $G \leftarrow $ Propagating-Equivalent-Nodes($G$)\;
  \Return $G$\;
  \caption{Preprocessing}\label{alg:preprocess}
\end{algorithm}

\paragraph{Merge Nodes with Equal Inputs. }
If multiple AIG nodes $\ell_i, \ell_j$ have the same inputs $a, b$,
we can express one gate polynomial using the other. For instance, in
our running Example~\ref{example:gate} the nodes
$\ell_{24} = \ell_{22}\ell_{16}$ and $\ell_{26}=(1-\ell_{22})(1-\ell_{16})$ would be
such a set of AIG nodes.

Every gate polynomial of an AIG node has degree two, and the quadratic
term is the product of the input nodes. Hence, the non-linear term
in those gate polynomials that have the same inputs is the same.
We remove the non-linear term of the topologically larger polynomial
by adding or subtracting the smaller polynomial.  For instance, we
derive $\ell_{26} - \ell_{24} + \ell_{22} + \ell_{16} - 1$.

Furthermore, assume two gates $\ell_i$ and $\ell_j$ both have input variables $a, b$.
 If at least one input has a different polarity in
$\ell_i$ and $\ell_j$, we immediately can derive that the product $\ell_i\ell_j$
is equal to zero.  To see this, let  $\ell_i - \bar{a}\bar{b}$, $\ell_j - \hat{a}\hat{b}$
be the corresponding gate polynomials, where $\bar{a}$ and $\hat{a}$
represent the polarity of~$a$.  We have
$\ell_i\ell_j = \bar{a}\bar{b}\hat{a}\hat{b} = 0$, since
$(\bar{a} = 1-\hat{a}) \lor (\bar{b} = 1-\hat{b})$ holds.  Thus, we
can always remove the term $\ell_i\ell_j$ in a possible parent node, for
instance the monomial $\ell_{26}\ell_{24}$ in $\ell_{28}$ in
Example~\ref{example:gate} can be removed.

\paragraph{Eliminate Positive Nodes. }
In this step we eliminate nodes which are only non-negated inputs to
other nodes in the graph. This heuristic was already
considered in~\cite{KaufmannBeameBiereNordstrom-DATE22} to introduce a
possible sharing of nodes. Since this heuristic is only applied on
positive inputs, we can simply replace every occurrence of the node by
the corresponding tail in the gate
polynomial of the parent node.  This will increase the degree of the
parent polynomial, but will not increase the number of terms.
We can check whether parts of the new
tail term of the parent are equal to the tail term of another gate
polynomial. If yes, we can reduce the tail term and include the
leading term. This will decrease the temporal increase of the
polynomial degree and furthermore will impose a node sharing which will
be useful in later Gröbner basis computations.
For instance, consider polynomials $f-da, e-ca, d-cb$. We can derive
  $f-cba = f-eb$.

\paragraph{Propagating Equivalent Nodes. } If at any point in the
rewriting we derive a linear polynomial of the form $\ell_i - \ell_j$ or
$\ell_i + \ell_j -1$ we know that
$\ell_i$ is equal to either~$\ell_j$ or to its negation $1-\ell_j$.
We propagate this information by eliminating the topologically larger node  $\ell_i$ from the
polynomial encoding. We choose to eliminate~$\ell_i$ and not $\ell_j$ in
order to not mess up the reverse topological term ordering for parent
nodes of $\ell_j$.  Propagation of equivalent nodes may not directly lead
to linear gate polynomials, but helps to reduce the overall number of
variables.

\subsection{Linear Reduction}\label{ssec:linearrewr}

After preprocessing we repeatedly rewrite the linearized specification by the polynomial $p$ in the Gröbner basis that has the same leading monomial as the specification (line 5 in Alg.~\ref{alg:specred}). For doing so, we need to linearize $p$. The pseudo-code is listed in Alg.~\ref{alg:linearizeviagb}.
By Theorem~\ref{thm:linGB}, we know that a full $\ldrl$-Gröbner basis of a circuit $C$  must contain a linear polynomial $p_{\lin}$ with the same leading monomial as~$p$. If this condition is not met, then
the circuit does not satisfy the specification, as we cannot further reduce $\spec$.

\begin{algorithm}[t]\small
  \SetKwInOut{Input}{Input} \SetKwInOut{Output}{Output}
  \SetKwComment{Comment}{\qquad \qquad\qquad$\triangleright$}{}
  \SetKw{And}{and}

  \Input{Polynomial $p$, Polynomial system $G$ }
  \Output{Linear polynomial $p_{\lin}$ or $\emptyset$}

  $v \leftarrow \lm(p); d \leftarrow 3$ \;
  \While{$d \leq \dist(v)$}{
    $C_v \leftarrow \{v\} \cup \{\text{Children-up-to-Distance}(v, d)\} \cup
    \{\text{Siblings}(v)\}$\;

    $C_v \leftarrow C_v \cup \{\text{Parents}(C_v)\}$\;
    $G_v \leftarrow  \text{Gate-Polynomials}(C_v, G) \cup \text{Boolean-Input-Polynomials}(C_v)$\;

    $G_{\DRL} \leftarrow$ Compute-$\ldrl$-Gröbner-Basis ($G_v$)\;
    \lIf{$\exists p_{\lin} \in G_{\DRL}$ such that $\deg(p_{\lin}) = 1 \land \lm(p_{\lin}) = v$}{\Return $p_{\lin}$}
    $d \leftarrow d + 1$\;
  }

  \Return $\emptyset$\;

  \caption{Linearize-Single-Polynomial}\label{alg:linearizeviagb}
\end{algorithm}

However, we do not want to compute a full $\ldrl$ Gröbner basis.
Our goal is to make the Gröbner basis just big enough such that it contains a
linear polynomial $p_{\lin}$ with leading term~$v$.
Let $v = \lm(p)$.
We aim to compute a Gröbner basis w.r.t. a $\ldrl$-ordering for a sub-circuit $C_{v,d}$ of $C$.
The sub-circuit $C_{v,d}$ is constructed by including $v$ and all children nodes of $v$ up to a maximum distance $d$.  Initially, we set $d=3$. The motivation for this threshold is that our preprocessing techniques already generates most linear polynomials detectable with $d=2$. On the other hand we do not want to start with a larger value for $d$ to keep the initial Gröbner basis computation as small as possible.
If we encounter a child node that already has a linear polynomial representation, we do not further add its children. This allows us to avoid
unnecessary computations by excluding parts of the circuit that have
already been simplified.
Additionally, we include all smaller
sibling nodes of $v$. Siblings are nodes that share at least one child
with $v$. Moreover, we collect all parent nodes whose children are already
included in the collected set of nodes. This ensures that all relevant
dependencies in the sub-circuit are captured.
This set of nodes represents the part of the circuit on which we
will compute a local $\ldrl$-Gröbner basis.

If this local Gröbner basis does not contain the expected linear
polynomial, it suggests that the sub-circuit $C_{v,d}$ is insufficient to
capture the desired behavior. In such cases, we repeat the process for the sub-circuit $C_{v,d+1}$, where we increase the distance $d$ to add more nodes.
Theoretically it would be very beneficial to cache the calls and reuse the computed Gröbner basis for $C_{v,d}$ for the Gröbner basis computation of $C_{v,d+1}$, however in practice most available Gröbner basis engines cannot exploit that a subset of the inputs is already  Gröbner basis.

We continue with the iterative process of increasing $d$ until either a linear polynomial is found, or, in the worst case we have computed a
full $\ldrl$-Gröbner basis for all gate polynomials that are topologically smaller
than~$v$. If we still did not find a linear polynomial at this point, we know that the circuit is incorrect. This follows from Theorem~\ref{thm:linGB}.

While our approach guarantees the completeness of the verification
process, it comes with a practical limitation: computational
complexity. If the sub-circuit grows too large (i.e., if too many
nodes need to be added to $C_v$), the computation of the  $\ldrl$-Gröbner basis becomes
infeasible in practice.

\section{Experimental Evaluation}\label{sec:experiments}

We evaluate our proposed approach on a set of multiplier benchmarks for different input bit-widths $n$.
For all the circuits we have $\spec = \sum_{i=0}^{2n} 2^is_i - (\sum_{i=0}^n 2^ia_i)\cdot(\sum_{i=0}^n 2^ib_i)$, hence choose $\KK = \QQ$.
Since all the leading coefficients of the gate polynomials are 1, the computation will stay in the ring $\ZZ[X] \subseteq \QQ[X]$~\cite{KaufmannBiereKauers-FMCAD19}.

\subsection{Implementation}\label{sec:implementation}

We implement Alg.~\ref{alg:specred} in our tool~\linGB~\cite{multiling,multiling-gh}, written in C++.
We employ the following features:
\begin{itemize}
\item
\linGB uses the polynomial arithmetic module from \amulettwo~\cite{KaufmannBiere-TACAS21}, which is targeted towards polynomial
arithmetic where the variables represent Boolean values and the
coefficients are integer values. In particular, the arithmetic engine
automatically includes reasoning over the Boolean input polynomials,
by reducing exponents, i.e., it calculates
$x\cdot x = x$ internally.
\item We sort the variables based on their minimum distance to the primary inputs to sort all extension variables next to the primary inputs, which
gave us better practical results than the column-wise variable order from \amulettwo.
\item As a consequence of the row-wise order, we do not apply an incremental column-wise reduction algorithm~\cite{KaufmannBiereKauers-FMSD19}, but
  rewrite the complete specification.
\item
For computing the $\ldrl$-Gröbner basis, we use the \msolve~\cite{berthomieu} library.
  Since \msolve is designed for general purposes, we have to explicitly
  provide the Boolean input polynomials.
\item If the distance of a node to the primary inputs is below six and the linearization of the individual polynomial fails, we switch to non-linear rewriting as a fall-back option. This threshold allows us to capture Booth encoding in our multiplier benchmarks. We empirically noticed that the linearization of Booth encodings requires a rather large $\ldrl$-Gröbner basis and it is computationally cheaper to use non-linear rewriting instead. However, switching to non-linear rewriting leads to a non-linear intermediate reduction result, meaning that we have to use non-linear rewriting also for the remainder of the circuit to maintain completeness.
\item In contrast to \amulettwo, we do not support proof logging in
  \linGB at the moment, as we have not yet instrumented \msolve to produce proofs in the PAC~\cite{KaufmannFleuryBiereKauers-FMSD22} format. This missing implementation is part of future work.
\end{itemize}


\subsection{Setup}

We run our experiments on a  Intel i7-1260P CPU.  The time is listed in rounded seconds (wall-clock time). We set the time limit to \SI{300}{\s} and the memory limit to \SI{10000}{\mega\byte}.
We compare \linGB against the algebraic approaches of  \amulettwo~\cite{KaufmannBiere-TACAS21}, \teluma~\cite{KaufmannBeameBiereNordstrom-DATE22}, and \dynphaseorderopt (\dpoo)~\cite{KonradScholl-FMCAD24}.
The tools of related works~\cite{LiuLiaoHuangZhenYuan-DATE24} and~\cite{LiLiYuFujitaJiangHa-TCAD24} are not publicly available.

\paragraph{Benchmarks. }
We evaluate our approach on integer multiplier circuits.
Multipliers consist of three main
components: partial product generation (PPG), partial product
accumulation (PPA), and a final-stage adder (FSA). Each component has
optimized architectures to reduce space and delay.

Two encodings are frequently used for PPG: simple AND-gate-based generation or
Booth encoding.  In
the former case, every partial product $a_ib_j$ is explicitly computed, hence we do not require extension variables in our approach. For Booth encoding, we require extension variables as the
partial products are internally combined.  During PPA, partial
products are accumulated, with the final two layers summed in the FSA.

In structured circuits, PPG, PPA, and FSA are clearly defined,
benefiting tools like \amulettwo and \teluma that require a clear cut
between PPA and FSA to simplify the FSA. In synthesized
circuits, gates are merged and rewritten to optimize the circuit, which
blurs these component boundaries, and complicates direct verification.
We consider two sets of benchmarks:
\begin{itemize}
\item \emph{aoki-multipliers}~\cite{aokipaper}: This set of benchmarks is generated by
combining different architectures for PPG, PPA, and FSA
\footnote{\tiny PPG: simple (sp), Booth encoding (bp);
PPA: Array                         (ar),
Wallace tree                  (wt),
Balanced delay tree           (bd),
Overturned-stairs tree        (os),
Dadda tree                    (dt),
(4;2) compressor tree         (ct),
(7,3) counter tree            (cn),
Red. binary addition tree     (ba);
FSA:
Ripple-carry                  (rc),
  Carry look-ahead              (cl),
  Ripple-block carry look-ahead (rb),
  Block carry look-ahead        (bc),
  Ladner-Fischer                (lf),
  Kogge-Stone                   (ks),
  Brent-Kung                    (bk),
  Han-Carlson                   (hc),
  Conditional sum               (cn),
  Carry select                  (cs),
  Carry-skip fix size           (csf),
  Carry-skip var. size          (csv)}
, yielding  192 non-synthesized multiplier architectures with an input bit-width 64.
\item \emph{optimized \abc multipliers}~\cite{abc}: We generate multipliers in \abc consisting of a simple PPG, an array PPA and a ripple-carry adder as FSA for bit-widths 32, 64, and 128, and optimize all of them within \abc using five
  different types of standard synthesis scripts:
  \emph{resyn}, \emph{resyn2}, \emph{resyn3},
  \emph{dc2}. We include a \emph{complex} script that
  combines several synthesis techniques\footnote{\tiny\texttt{-c "logic;
      mfs2 -W 20; ps; mfs; st; ps; dc2 -l; ps; resub -l -K 16 -N 3 -w
      100; ps; logic; mfs2 -W 20; ps; mfs; st; ps; iresyn -l; ps;
      resyn; ps; resyn2; ps; resyn3; ps; dc2 -l; ps;"}}.
      We include these 15 optimized benchmarks to demonstrate the robustness of our presented approach. Optimized benchmarks are particularly challenging as they cannot be fully decomposed into their building blocks.

\end{itemize}
All benchmarks model correct multipliers, i.e., the circuits fulfill the specification.

\subsection{Results}

  \begin{table}[tb]
    \scriptsize
    \centering
    \begin{tabular}{r|l|r||r|r|r||r|r|r|r|r|rr|rr|rr}
      \multicolumn{3}{c||}{\abc-benchmarks} & \multicolumn{3}{c||}{Related work}  & \multicolumn{2}{c|}{Preprocess} &\multicolumn{2}{c|}{Time (\SI{}{\s})}  &     \multicolumn{7}{c}{\msolve Calls}   \\
      \hline
      $n$ & Synth &Nodes & \cite{KaufmannBeameBiereNordstrom-DATE22} & \cite{KaufmannBiere-TACAS21} & \cite{KonradScholl-FMCAD24}  & MergedN & PosN & Total & \msolve & \# & \multicolumn{2}{c|}{$d=3$ (\SI{}{\s})} & \multicolumn{2}{c|}{$d=4$ (\SI{}{\s})}& \multicolumn{2}{c}{$d=5$ (\SI{}{\s})}\\
      \hline
      32 &    resyn &   7840 &  0.1 &  TO & 0.2 &   1948 &      1 &   1.7 &   0.2 &      7 &    4  &  0.03 &   2  &  0.03 &   1  &  0.04 \\
      32 &   resyn2 &   7840 &  0.1 &  TO & 0.3 &   1948 &      1 &   1.3 &   0.2 &      6 &    4  &  0.03 &   1  &  0.03 &   1  &  0.04 \\
      32 &   resyn3 &   7840 &  0.1 & 0.01 & 0.3 &   1952 &      0 &   0.8 &   0.0 &      0 &    0  &  &   0  &   &   0  &   \\
      32 &      dc2 &   7840 &  0.1 & 0.01 & 0.2 &   1952 &      0 &   1.0 &  0.0 &      0 &    0  &   &   0  &   &   0  &   \\
      32 &  comp &   7839 &   EE &  TO & 0.2 &   1948 &      0 &   1.5 &   0.2 &      6 &    4  &  0.03 &   1  &  0.03 &   1  &  0.04 \\
      \hline
      64 &    resyn &  32064 &  0.3 &  TO & 1.0 &   7996 &      1 &  11.0 &   0.2 &      7 &    4  &  0.03 &   2  &  0.03 &   1  &  0.04 \\
      64 &   resyn2 &  32064 &  0.2 &  TO & 1.0 &   7996 &      1 &  11.7 &   0.2 &      6 &    4  &  0.03 &   1  &  0.03 &   1  &  0.04 \\
      64 &   resyn3 &  32064 &  0.3 & 0.2 & 1.0 &   8000 &      0 &  11.2 &  0.0 &      0 &    0  &   &   0  &   &   0  &   \\
      64 &      dc2 &  32064 &  0.2 & 0.3 & 1.0 &   8000 &      0 &  11.1 &  0.0 &      0 &    0  &   &   0  &   &   0  &   \\
      64 &  comp &  32063 &   EE &  TO & 1.0 &   7996 &      0 &  12.2 &   0.2 &      6 &    4  &  0.03 &   1  &  0.03 &   1  &  0.04 \\
      \hline
     128 &    resyn & 129664 &  1.2 &  TO & 5.7 &  32380 &      1 & 228.6 &   0.2 &      7 &    4  &  0.03 &   2  &  0.04 &   1  &  0.04 \\
     128 &   resyn2 & 129664 &  1.2 &  TO & 6.4 &  32380 &      1 & 232.7 &   0.2 &      6 &    4  &  0.03 &   1  &  0.03 &   1  &  0.04 \\
     128 &   resyn3 & 129664 &  1.2 &  TO & 7.7 &  32384 &      0 & 228.9 &   0.0 &      0 &    0  &   &   0  &   &   0  &   \\
     128 &      dc2 & 129664 &  1.1 &  TO & 6.6 &  32384 &      0 & 229.4 &   0.0 &      0 &    0  &   &   0  &   &   0  &   \\
     128 &  comp & 129663 &   EE &  TO & 5.8 &  32380 &      0 & 230.4 &   0.2 &      6 &    4  &  0.04 &   1  &  0.04 &   1  &  0.04 \\

    \end{tabular}
    \vspace{1ex}
    \caption{Results on $n$-bit \abc benchmarks. \qquad
      TO $\coloneqq {}>\SI{300}{\s}$;
      EE $\coloneqq$ segfault.}\label{tbl:abc}\vspace{-5ex}
    \end{table}

The results for the optimized \abc multipliers are shown in Table~\ref{tbl:abc}.
The heuristics of \amulettwo and \teluma are not robust for these benchmarks and
produce time outs. \dpoo and our tool \linGB are both able to solve all benchmarks within the
time limit. We provide statistics on \linGB and show how often nodes with equal inputs are merged (``MergedN''), the number
of eliminated positive nodes (``PosN''). We explicitly measure the time that is required by \msolve. This time is included in the total computation time.
Additionally we provide statistics on the \msolve calls and list the total number (``\#''), as well as the number of calls for each depth $d$ and the \emph{average} computation time per depth.
For ``resyn3'' and ``dc2'' everything could be linearized via merging nodes.
This is in high contrast to the aoki benchmarks, see Table~\ref{tbl:aoki} and we
believe this is due to dense structure of the \abc graphs, which involve a higher sharing of nodes leading to more node pairs with equal children.

\begin{figure}[tb]
  \centering
  \begin{minipage}{0.50\textwidth}
      \centering
      \includegraphics[height=4.5cm]{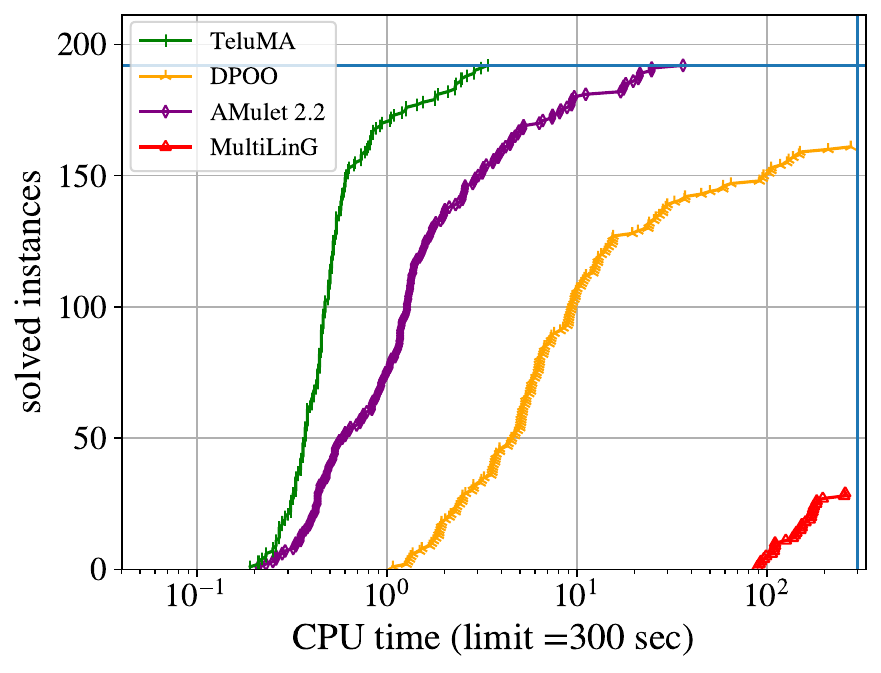}
      \caption{Results of aoki-benchmarks}
      \label{fig:exp_aoki}
  \end{minipage}
  \hfill
  \begin{minipage}{0.45\textwidth}
      \centering
      \includegraphics[height=4.5cm]{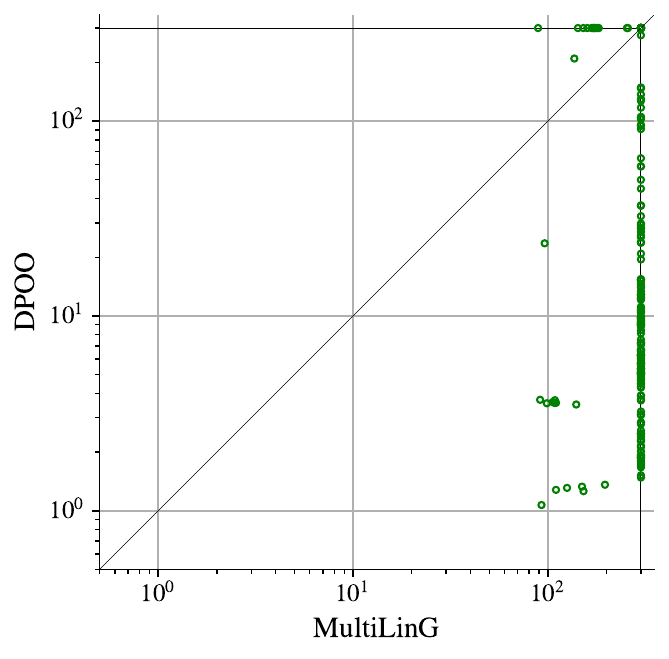}
      \caption{\linGB versus \dpoo}\label{fig:scatter}
  \end{minipage}
  \vspace{-2ex}
\end{figure}

\begin{table}[tb]
  \scriptsize
  \centering
  \begin{tabular}{l|l||r|r||r|r||r|rr|rr|rr|rr||r}
    \multicolumn{2}{c||}{Benchmarks} &  \multicolumn{2}{c||}{Preprocess} & \multicolumn{2}{c||}{Time (\SI{}{\s})}  &     \multicolumn{9}{c||}{\msolve Calls} & Nlin \\
    \hline
    Name & Nodes & Mrg & PosN & Total & \msolve (\%) & \# & \multicolumn{2}{c|}{d=3 (\SI{}{\s})} & \multicolumn{2}{c|}{d=4 (\SI{}{\s})}& \multicolumn{2}{c|}{d=5 (\SI{}{\s})} & \multicolumn{2}{c||}{d=6 (\SI{}{\s})} & \%\\
    \hline
    sparrc & 48000 & 8000 & 3968 & 92.7 & 79.5 (86.0) & 3968 & 3968 & 0.02 & 0 &  & 0 &  & 0 &  & 0.0  \\
    sparcl & 53733 & 8000 & 7749 & 96.2 & 79.4 (82.8) & 3969 & 3969 & 0.02 & 0 &  & 0 &  & 0 &  & 0.0  \\
    spwtrc & 49312 & 8332 & 3964 & 152.3 & 119.2 (78.4) & 4004 & 3990 & 0.03 & 14 & 0.03 & 0 &  & 0 &  & 0.0  \\
    spwtcl & 68747 & 8332 & 16845 & 176.2 & 122.6 (69.8) & 4003 & 3990 & 0.03 & 13 & 0.03 & 0 &  & 0 &  & 0.0  \\
    spbdrc & 49116 & 8281 & 3966 & 110.0 & 87.3 (79.8) & 3968 & 3967 & 0.02 & 1 & 0.02 & 0 &  & 0 &  & 0.0  \\
    spbdcl & 69240 & 8281 & 17305 & 142.5 & 101.5 (71.7) & 3966 & 3966 & 0.03 & 0 &  & 0 &  & 0 &  & 0.0  \\
    sposrc & 49392 & 8350 & 3966 & 125.3 & 97.2 (78.1) & 3968 & 3967 & 0.02 & 1 & 0.02 & 0 &  & 0 &  & 0.0  \\
    sposcl & 71291 & 8350 & 18485 & 151.8 & 105.5 (69.9) & 3966 & 3966 & 0.03 & 0 &  & 0 &  & 0 &  & 0.0  \\
    spdtrc & 48000 & 8000 & 3968 & 149.6 & 116.3 (78.1) & 3968 & 3968 & 0.03 & 0 &  & 0 &  & 0 &  & 0.0  \\
    spdtcl & 71000 & 8000 & 19219 & 174.6 & 115.1 (66.1) & 3968 & 3968 & 0.03 & 0 &  & 0 &  & 0 &  & 0.0  \\
    spctrc & 41248 & 8194 & 208 & 196.2 & 162.9 (83.3) & 7763 & 3996 & 0.02 & 1884 & 0.02 & 1883 & 0.03 & 0 &  & 0.0  \\
    spctcl & 62428 & 8194 & 14249 & 257.4 & 171.9 (66.9) & 7761 & 3995 & 0.02 & 1883 & 0.02 & 1883 & 0.03 & 0 &  & 0.0  \\
    spcnrc & 47236 & 8076 & 3060 & 136.6 & 105.7 (77.6) & 3037 & 3037 & 0.04 & 0 &  & 0 &  & 0 &  & 0.0  \\
    spcncl & 70236 & 8076 & 18311 & 173.3 & 111.3 (64.7) & 3037 & 3037 & 0.04 & 0 &  & 0 &  & 0 &  & 0.0  \\
    bparrc & 38311 & 6794 & 2137 & 98.7 & 38.8 (39.6) & 2112 & 2108 & 0.02 & 2 & 0.02 & 1 & 0.05 & 1 & 0.05 & 38.2  \\
    bpwtrc & 37315 & 6550 & 2134 & 91.2 & 43.6 (47.9) & 2139 & 2125 & 0.02 & 12 & 0.02 & 1 & 0.06 & 1 & 0.06 & 39.3  \\
    bpwtcl & 57556 & 6550 & 15589 & 159.1 & 46.6 (29.5) & 2137 & 2124 & 0.02 & 11 & 0.03 & 1 & 0.05 & 1 & 0.06 & 25.5  \\
    bpbdrc & 37365 & 6561 & 2135 & 110.0 & 49.0 (44.9) & 2110 & 2106 & 0.02 & 2 & 0.03 & 1 & 0.05 & 1 & 0.05 & 39.2  \\
    bpbdcl & 58309 & 6561 & 16058 & 181.0 & 44.7 (24.8) & 2109 & 2106 & 0.02 & 1 & 0.03 & 1 & 0.05 & 1 & 0.06 & 25.1  \\
    bposrc & 37459 & 6584 & 2136 & 105.8 & 41.3 (39.4) & 2111 & 2107 & 0.02 & 2 & 0.03 & 1 & 0.06 & 1 & 0.06 & 39.1  \\
    bposcl & 58759 & 6584 & 16296 & 170.9 & 49.5 (29.1) & 2109 & 2106 & 0.02 & 1 & 0.04 & 1 & 0.06 & 1 & 0.06 & 24.9  \\
    bpdtrc & 36044 & 6237 & 2114 & 108.5 & 47.2 (43.6) & 2087 & 2084 & 0.02 & 1 & 0.03 & 1 & 0.05 & 1 & 0.05 & 40.6  \\
    bpdtcl & 59169 & 6237 & 17489 & 167.5 & 52.4 (31.5) & 2087 & 2084 & 0.03 & 1 & 0.03 & 1 & 0.06 & 1 & 0.06 & 24.7  \\
    bpctrc & 32951 & 6428 & 193 & 139.7 & 101.9 (73.1) & 3991 & 2112 & 0.02 & 939 & 0.02 & 939 & 0.04 & 1 & 0.04 & 44.5  \\
    bpctcl & 54251 & 6428 & 14353 & 254.2 & 115.3 (45.6) & 3991 & 2112 & 0.02 & 939 & 0.03 & 939 & 0.04 & 1 & 0.06 & 27.0  \\
    bpcnrc & 35557 & 6300 & 1616 & 89.0 & 39.1 (44.0) & 1609 & 1606 & 0.02 & 1 & 0.03 & 1 & 0.05 & 1 & 0.05 & 41.2  \\
    bpcncl & 58682 & 6300 & 16991 & 170.1 & 50.3 (29.6) & 1609 & 1606 & 0.03 & 1 & 0.03 & 1 & 0.06 & 1 & 0.06 & 24.9  \\
    bpbarc & 38141 & 6590 & 2330 & 108.6 & 46.8 (43.2) & 2322 & 2314 & 0.02 & 6 & 0.03 & 1 & 0.05 & 1 & 0.05 & 38.7  \\
    bpbacl & 61768 & 6590 & 18080 & 182.2 & 54.9 (30.3) & 2321 & 2313 & 0.02 & 6 & 0.02 & 1 & 0.04 & 1 & 0.04 & 23.9  \\

  \end{tabular}
  \vspace{1ex}
  \caption{Results on \emph{solved} aoki benchmarks.}\label{tbl:aoki}\vspace{-5ex}
  \end{table}

Figure~\ref{fig:exp_aoki} shows the results on the aoki-benchmarks. Both, \teluma and \amulettwo, are able to solve the complete benchmark set.
\dpoo solves 163 out of 192 benchmarks, whereas our approach is only able to solve 29 benchmarks. 106 benchmarks exceed the memory limit and 57 benchmarks exceed the time limit.
Details on the solved instances are given in Table~\ref{tbl:aoki}.

Although the number of solved instances is low for \linGB, we are able to solve 13 benchmarks that \dpoo does not cover, see Figure~\ref{fig:scatter}.
All those instances use a carry-lookahead adder  (cl) as FSA, which includes sequences of OR-gates that lead to a monomial blow-up when rewritten.
\amulettwo and \teluma solve these benchmarks using either a SAT solver or polynomial rewriting to replace the FSA with an equivalent ripple-carry adder.
In our approach these circuits benefit from rewriting positive nodes.

Table~\ref{tbl:aoki} provides additional insights and shows that for multipliers using a simple PPG between 65--86\% of the computation time is spent in \msolve. For multipliers using a Booth encoding (bp) this percentage is lower, as we switch to non-linear rewriting during the reduction. Column ``Nlin'' provides the percentage of ``Nodes'' that are reduced using non-linear rewriting.

\emph{Summarizing, \amulettwo and \teluma are highly efficient on the structured circuits but are not robust on optimized benchmarks.
\dpoo and \linGB are both robust and complement each other on complex multiplier designs.
Hence our proposed approach is a valuable addition to the algebraic verification landscape and will be even more powerful when combined with existing methods.}

\section{Conclusion}\label{sec:conclusion}

In this paper we have presented a novel technique to verify directed
acyclic graphs using computer algebra.  Our first contribution is a
theoretical theorem that shows how we can perform the ideal membership
test of a specification polynomial using only linear polynomial
operations.  Secondly, we discuss how we can apply this theorem in
practice to overcome the overhead of computing a full
Gröbner basis. We present a technique that incrementally computes
Gröbner bases for small sub-graphs to extract the linear information
of the polynomials. We have demonstrated the potential of our approach
on a set of multiplier circuits that have been challenging to verify
so far.

In the future we aim to turn the black-box Gröbner basis approach
into a white-box and explore how we can derive the linear polynomials
without the computation of a
full
Gröbner basis. We also
envision equivalence checking as a potential application, as
this restricts the computation to Boolean polynomials.

\begin{credits}
  \subsubsection{\ackname}
  This research was supported by the Austrian Science Fund (FWF)
  [10.55776/ESP666], the joint ANR-FWF
  ANR-19-CE48-0015 \textsc{ECARP}
  and
  ANR-22-CE91-0007 \textsc{EAGLES}
  projects,
  the ANR-19-CE40-0018 \textsc{De Rerum Natura}
  project,
  and
  grants
  DIMRFSI 2021-02–C21/1131 of the Paris Île-de-France Region and
  FA8665-20-1-7029 of
  the EOARD-AFOSR.
\end{credits}
%
%
%
\bibliographystyle{splncs04} \bibliography{degboundedgb}

\begin{thebibliography}{10}
\providecommand{\url}[1]{\texttt{#1}}
\providecommand{\urlprefix}{URL }
\providecommand{\doi}[1]{https://doi.org/#1}

\bibitem{abc}
{Berkeley Logic Synthesis and Verification Group}: {ABC: A System for
  Sequential Synthesis and Verification}.
  \url{http://www.eecs.berkeley.edu/~alanmi/abc/} (2019), bitbucket Version
  1.01

\bibitem{BerthomieuNS2022}
Berthomieu, J., Neiger, V., Safey El~Din, M.: {Faster Change of Order Algorithm
  for Gr\"{o}bner Bases under Shape and Stability Assumptions}. In: ISSAC. pp.
  409--418. ACM (2022). \doi{10.1145/3476446.3535484}

\bibitem{berthomieu}
Berthomieu, J., Eder, C., {Safey El Din}, M.: {msolve: A Library for Solving
  Polynomial Systems}. In: {ISSAC}. pp. 51--58. {ACM} (2021).
  \doi{10.1145/3452143.3465545}

\bibitem{Biere-SAT-Competition-2016-benchmarks}
Biere, A.: {Collection of Combinational Arithmetic Miters Submitted to the {SAT
  Competition 2016}}. In: {SAT Competition} 2016. Dep. of Computer Science
  Report Series B, vol. B-2016-1, pp. 65--66. University of Helsinki (2016)

\bibitem{Buchberger65}
Buchberger, B.: {Ein Algorithmus zum Auffinden der Basiselemente des
  Restklassenringes nach einem nulldimensionalen Polynomideal}. Ph.D. thesis,
  University of Innsbruck (1965)

\bibitem{buchberger10}
Buchberger, B., Kauers, M.: Gr{\"o}bner basis. Scholarpedia  \textbf{5}(10),
  ~7763 (2010), \url{http://www.scholarpedia.org/article/Groebner_basis}

\bibitem{ChenBryant-DAC95}
Chen, Y., Bryant, R.E.: {Verification of Arithmetic Circuits with Binary Moment
  Diagrams}. In: DAC. pp. 535--541. {ACM} (1995). \doi{10.1145/217474.217583}

\bibitem{CoxLittleOShea-Book07}
Cox, D., Little, J., O'Shea, D.: {Ideals, Varieties, and Algorithms}.
  Springer-Verlag New York (1997)

\bibitem{Faugere1999}
Faug{\`e}re, J.{\relax Ch}.: {A New Efficient Algorithm for Computing
  {Gr\"obner} bases (F4)}. Journal of Pure and Applied Algebra
  \textbf{139}(1),  61--88 (1999). \doi{10.1016/S0022-4049(99)00005-5}

\bibitem{FaugereGHR2014}
Faug\`{e}re, J.{\relax Ch}., Gaudry, P., Huot, L., Renault, G.: {Sub-Cubic
  Change of Ordering for Gr\"{o}bner Basis: A Probabilistic Approach}. In:
  ISSAC. pp. 170--177. ACM (2014). \doi{10.1145/2608628.2608669}

\bibitem{FaugereGLM1993}
Faug{\`e}re, J.{\relax Ch}., Gianni, P., Lazard, D., Mora, T.: Efficient
  {C}omputation of {Z}ero-dimensional {G}r{\"o}bner {B}ases by {C}hange of
  {O}rdering. J.\ Symbolic Comput.  \textbf{16}(4),  329--344 (1993).
  \doi{10.1006/jsco.1993.1051}

\bibitem{FaugereM2017}
Faug{\`e}re, J.{\relax Ch}., Mou, C.: {Sparse FGLM algorithms}. {Journal of
  Symbolic Computation}  \textbf{80}(3),  538--569 (2017).
  \doi{10.1016/j.jsc.2016.07.025}

\bibitem{aokipaper}
Homma, N., Watanabe, Y., Aoki, T., Higuchi, T.: {Formal Design of Arithmetic
  Circuits Based on Arithmetic Description Language}. {IEICE} Trans. Fundam.
  Electron. Commun. Comput. Sci.  \textbf{89-A}(12),  3500--3509 (2006).
  \doi{10.1093/IETFEC/E89-A.12.3500}

\bibitem{multiling-gh}
Kaufmann, D.: {MultiLinG} (2025),
  \url{https://www.github.com/d-kfmnn/multiling}

\bibitem{KaufmannBeameBiereNordstrom-DATE22}
Kaufmann, D., Beame, P., Biere, A., Nordstr{\"{o}}m, J.: Adding dual variables
  to algebraic reasoning for gate-level multiplier verification. In: {DATE}.
  pp. 1431--1436. {IEEE} (2022). \doi{10.23919/DATE54114.2022.9774587}

\bibitem{multiling}
Kaufmann, D., Berthomieu, J.: {MultiLinG} - {Extracting Linear Relations from
  Gröbner Bases for Formal Verification of And-Inverter Graphs (Artifact)}
  (2025). \doi{10.5281/zenodo.14609934}

\bibitem{KaufmannBiere-TACAS21}
Kaufmann, D., Biere, A.: Amulet 2.0 for verifying multiplier circuits. In:
  {TACAS} {(2)}. LNCS, vol. 12652, pp. 357--364. Springer (2021).
  \doi{10.1007/978-3-030-72013-1\_19}

\bibitem{KaufmannBiereKauers-FMCAD19}
Kaufmann, D., Biere, A., Kauers, M.: {Verifying Large Multipliers by Combining
  {SAT} and Computer Algebra}. In: {FMCAD} 2019. pp. 28--36. {IEEE} (2019).
  \doi{10.23919/FMCAD.2019.8894250}

\bibitem{KaufmannBiereKauers-FMSD19}
Kaufmann, D., Biere, A., Kauers, M.: {Incremental Column-wise verification of
  arithmetic circuits using computer algebra}. Formal Methods Syst. Des.
  \textbf{56}(1),  22--54 (2020). \doi{10.1007/S10703-018-00329-2}

\bibitem{KaufmannFleuryBiereKauers-FMSD22}
Kaufmann, D., Fleury, M., Biere, A., Kauers, M.: {Practical Algebraic Calculus
  and Nullstellensatz with the Checkers Pacheck and Pastèque and
  Nuss-Checker}. Formal Methods Syst. Des.  \textbf{64}(1),  73--107 (2022).
  \doi{10.1007/s10703-022-00391-x}

\bibitem{KonradScholl-FMCAD24}
Konrad, A., Scholl, C.: {Symbolic Computer Algebra for Multipliers Revisited -
  It’s All About Orders and Phases}. In: {FMCAD} 2024. pp. 261--271. {TU Wien
  Academic Press} (2024). \doi{10.34727/2024/isbn.978-3-85448-065-5_32}

\bibitem{KonradSchollMahzoonGrosseDrechsler-FMCAD22}
Konrad, A., Scholl, C., Mahzoon, A., Gro{\ss}e, D., Drechsler, R.: Divider
  verification using symbolic computer algebra and delayed don't care
  optimization. In: {FMCAD}. pp. 1--10. {IEEE} (2022).
  \doi{10.34727/2022/ISBN.978-3-85448-053-2\_17}

\bibitem{Kuehlmann-TCAD2002}
Kuehlmann, A., Paruthi, V., Krohm, F., Ganai, M.: {Robust Boolean reasoning for
  equivalence checking and functional property verification.} IEEE TCAD
  \textbf{21}(12),  1377--1394 (2002). \doi{10.1109/TCAD.2002.804386}

\bibitem{LiLiYuFujitaJiangHa-TCAD24}
Li, R., Li, L., Yu, H., Fujita, M., Jiang, W., Ha, Y.: Refscat: Formal
  verification of logic-optimized multipliers via automated reference
  multiplier generation and sca-sat synergy. IEEE TCAD pp.~1--1 (2024).
  \doi{10.1109/TCAD.2024.3442987}

\bibitem{LiewNordstroem-FMCAD20}
Liew, V., Beame, P., Devriendt, J., Elffers, J., Nordström, J.: {Verifying
  Properties of Bit-vector Multiplication Using Cutting Planes Reasoning}. In:
  {FMCAD} 2020. FMCAD, vol.~1, pp. 194--204. {TU Vienna Academic Press} (2020).
  \doi{10.34727/2020/ISBN.978-3-85448-042-6\_27}

\bibitem{LiuLiaoHuangZhenYuan-DATE24}
Liu, H., Liao, P., Huang, J., Zhen, H.L., Yuan, M., Ho, T.Y., Yu, B.: Parallel
  gröbner basis rewriting and memory optimization for efficient multiplier
  verification. In: DATE. pp.~1--6 (2024).
  \doi{10.23919/DATE58400.2024.10546568}

\bibitem{LvKallaEnescu-TCAD13}
Lv, J., Kalla, P., Enescu, F.: {Efficient {Gr{\"{o}}bner} Basis Reductions for
  Formal Verification of {Galois} Field Arithmetic Circuits}. IEEE TCAD
  \textbf{32}(9),  1409--1420 (2013). \doi{10.1109/TCAD.2013.2259540}

\bibitem{MahzoonGrosseDrechsler-ICCAD18}
Mahzoon, A., Gro{\ss}e, D., Drechsler, R.: {PolyCleaner: Clean your Polynomials
  before Backward Rewriting to verify Million-gate Multipliers}. In: {ICCAD}
  2018. pp. 129:1 -- 129:8. {ACM} (2018). \doi{10.1145/3240765.3240837}

\bibitem{MahzoonGrosseSchollDrechsler-DATE20}
Mahzoon, A., Gro{\ss}e, D., Scholl, C., Drechsler, R.: Towards formal
  verification of optimized and industrial multipliers. In: {DATE}. pp.
  544--549. {IEEE} (2020). \doi{10.23919/DATE48585.2020.9116485}

\bibitem{MayrM1982}
Mayr, E.W., Meyer, A.R.: The complexity of the word problems for commutative
  semigroups and polynomial ideals. Advances in Mathematics  \textbf{46}(3),
  305--329 (1982). \doi{10.1016/0001-8708(82)90048-2}

\bibitem{NeigerS2020}
Neiger, V., Schost, {\'E}.: Computing syzygies in finite dimension using fast
  linear algebra. Journal of Complexity  \textbf{60},  101502 (2020).
  \doi{10.1016/j.jco.2020.101502}

\bibitem{Sharangpani1994StatisticalAO}
Sharangpani, H., Barton, M.L.: Statistical analysis of floating point flaw in
  the pentium processor (1994)

\bibitem{Temel-TACAS24}
Temel, M.: Vescmul: Verified implementation of s-c-rewriting for multiplier
  verification. In: {TACAS} {(1)}. LNCS, vol. 14570, pp. 340--349. Springer
  (2024). \doi{10.1007/978-3-031-57246-3\_19}

\end{thebibliography}

\newpage
\appendix
\section{Complete Gröbner basis for two-bit multiplier}%
\label{appendix:deggb-2mult}

\begin{lstlisting}[
  basicstyle=\footnotesize
]
gb[1]=l10-t00               gb[27]=l16*a1-l16
gb[2]=s0-l10                gb[28]=b0^2-b0
gb[3]=l12-t10               gb[29]=t10*b0-t10
gb[4]=l14-t01               gb[30]=t00*b0-t00
gb[5]=l18-l16+t01+t10-1     gb[31]=l16*b0-l16
gb[6]=l20+2*l16-t01-t10     gb[32]=b1^2-b1
gb[7]=s1-l20                gb[33]=t11*b1-t11
gb[8]=l22-t11               gb[34]=t10*b1-t11*b0
gb[9]=l24-l16               gb[35]=t01*b1-t01
gb[10]=l26-l24+l16+t11-1    gb[36]=t00*b1-t01*b0
gb[11]=l28+2*l24-l16-t11    gb[37]=l16*b1-l16
gb[12]=s2-l28               gb[38]=t11^2-t11
gb[13]=s3-l24               gb[39]=t10*t11-t11*b0
gb[14]=a0^2-a0              gb[40]=t01*t11-t11*a0
gb[15]=b0*a0-l10            gb[41]=t00*t11-l16
gb[16]=b1*a0-t01            gb[42]=l16*t11-l24
gb[17]=t01*a0-t01           gb[43]=t10^2-t10
gb[18]=t00*a0-t00           gb[44]=t01*t10-l16
gb[19]=l16*a0-l16           gb[45]=t00*t10-t10*a0
gb[20]=a1^2-a1              gb[46]=l16*t10-l16
gb[21]=b0*a1-t10            gb[47]=t01^2-t01
gb[22]=b1*a1-t11            gb[48]=t00*t01-t01*b0
gb[23]=t11*a1-t11           gb[49]=l16*t01+l16*t10+l20-t01-t10
gb[24]=t10*a1-t10           gb[50]=t00^2-t00
gb[25]=t01*a1-t11*a0        gb[51]=l16*t00-l16
gb[26]=t00*a1-t10*a0        gb[52]=l16^2-l16
\end{lstlisting}

\end{document}